\documentclass[letterpaper, 10 pt, conference]{ieeeconf}  % Comment this line out if you need a4paper

\IEEEoverridecommandlockouts                              % This command is only needed if 
\title{\LARGE \bf
Viability under Degraded Control Authority
}

\author{Hamza~El-Kebir$^{1}$, Richard Berlin$^{2}$, Joseph~Bentsman$^{3}$, and Melkior~Ornik$^{4}$%, Martin Ostoja-Starzewski$^{3}$, Richard Berlin$^{4}$, \\ Joseph~Bentsman$^{2}$, and Leonardo~P.~Chamorro$^{2}$% <-this % stops a space
\thanks{$^{1}$H. El-Kebir is with the Department
of Aerospace Engineering, University of Illinois Urbana-Champaign, Urbana,
IL, 61801 USA
{\tt\small elkebir2@illinois.edu}}%
\thanks{$^{2}$R. Berlin is with the Department of Trauma Surgery, Carle Hospital and the University of Illinois Urbana-Champaign, Urbana, IL 61801 USA.}%
\thanks{$^{3}$J. Bentsman is with the Department
of Mechanical Science and Engineering, University of Illinois Urbana-Champaign, Urbana,
IL, 61801 USA
{\tt\small jbentsma@illinois.edu}}%
\thanks{$^{4}$M. Ornik is with the Department
of Aerospace Engineering and the Coordinated Science Laboratory, University of Illinois Urbana-Champaign, Urbana,
IL, 61801 USA
{\tt\small mornik@illinois.edu}}%
\thanks{Research reported in this publication was supported by the National Institute of Biomedical Imaging and Bioengineering of the National Institutes of Health under award number R01EB029766. The content is solely the responsibility of the authors and does not necessarily represent the official views of the National Institutes of Health.}
}

\let\labelindent\relax

\usepackage{amsmath}
\usepackage{mathrsfs}
\usepackage{graphicx}
\usepackage[inline]{enumitem}

\usepackage{bm}
\renewcommand{\vec}[1]{\boldsymbol{\mathbf{#1}}}

\usepackage{amsthm}

\newtheorem{theorem}{Theorem}
\newtheorem{lemma}{Lemma}

\newtheorem{proposition}{Proposition}

\theoremstyle{definition}
\newtheorem{definition}{Definition}[section]
\newtheorem{problem}{Problem}
\newtheorem{assump}{Assumption}

\theoremstyle{remark}
\newtheorem{remark}{Remark}

\newtheorem{corollary}{Corollary}

\usepackage[T1]{fontenc}
\usepackage{calligra}

\usepackage{stix}

\usepackage{dsfont}
\usepackage{hyperref}
\usepackage[inline]{enumitem}
\usepackage{graphicx}
\usepackage{caption}
\usepackage{subcaption }
\usepackage{tabularx}
\usepackage{booktabs}

\newcommand{\transp}{^\mathsf{T}}
\newcommand{\inv}{^{-1}}

\usepackage{scalerel, stackengine}

\stackMath
\newcommand{\dhat}[1]{\ThisStyle{\setbox0=\hbox{$\SavedStyle#1$}%
  \stackengine{0pt}{\SavedStyle#1}{\SavedStyle\hspace{.2\ht0}%
  \hat{\vphantom{#1}}\kern\dimexpr2.2\LMpt+.7pt\relax\hat{\vphantom{#1}}}{O}{c}{F}{T}{L}}%
}

\newcommand{\dcheck}[1]{\ThisStyle{\setbox0=\hbox{$\SavedStyle#1$}%
  \stackengine{0pt}{\SavedStyle#1}{\SavedStyle\hspace{.2\ht0}%
  \check{\vphantom{#1}}\kern\dimexpr2.2\LMpt+.7pt\relax\check{\vphantom{#1}}}{O}{c}{F}{T}{L}}%
}

\newcommand{\hatcheck}[1]{\ThisStyle{\setbox0=\hbox{$\SavedStyle#1$}%
  \stackengine{0pt}{\SavedStyle#1}{\SavedStyle\hspace{.2\ht0}%
  \hat{\vphantom{#1}}\kern\dimexpr2.2\LMpt+.7pt\relax\check{\vphantom{#1}}}{O}{c}{F}{T}{L}}%
}

\newcommand{\checkhat}[1]{\ThisStyle{\setbox0=\hbox{$\SavedStyle#1$}%
  \stackengine{0pt}{\SavedStyle#1}{\SavedStyle\hspace{.2\ht0}%
  \check{\vphantom{#1}}\kern\dimexpr2.2\LMpt+.7pt\relax\hat{\vphantom{#1}}}{O}{c}{F}{T}{L}}%
}

\renewcommand{\vec}[1]{\boldsymbol{\mathbf{#1}}}

\usepackage{stmaryrd}

\begin{document}

\maketitle
\thispagestyle{empty}
\pagestyle{empty}

\begin{abstract}
In this work, we solve the problem of quantifying and mitigating control authority degradation in real time. Here, our target systems are controlled nonlinear affine-in-control evolution equations with finite control input and finite- or infinite-dimensional state. We consider two cases of control input degradation: finitely many affine maps acting on unknown disjoint subsets of the inputs and general Lipschitz continuous maps. These degradation modes are encountered in practice due to actuator wear and tear, hard locks on actuator ranges due to over-excitation, as well as more general changes in the control allocation dynamics. We derive sufficient conditions for identifiability of control authority degradation, and propose a novel real-time algorithm for identifying or approximating control degradation modes. 
%We also introduce the concept of a \emph{viabilizing map}, which remaps commanded control signals to viabilized signals that produce a minimally disturbed approximation of the commanded control signal after control authority degradation. We develop tight error bounds on the control degradation mode approximation and the viabilizing map, thereby producing efficiently computable robust control specifications. 
We demonstrate our method on a nonlinear distributed parameter system, namely a one-dimensional heat equation with a velocity-controlled moveable heat source, motivated by autonomous energy-based surgery.
\end{abstract}

%\begin{IEEEkeywords}
%distributed parameter system, Stefan problem, electrosurgery, electric discharge machining, exponential stability, modeling, Spalding theory.
%\end{IEEEkeywords}

% \fi
\IEEEpeerreviewmaketitle

\section{Introduction}

In control systems, fault detection and mitigation is key in ensuring prolonged safe operation in safety-critical environments \cite{Blanke2006}. Any physical system undergoes gradual degradation during its operational life cycle, for instance due to interactions with the environment or from within as a result of actuator wear and tear. Gradual degradation or impairment, as the name suggests, often reduces the performance of a system in cases when potential degradation modes were not taken into account during control synthesis. Fault tolerance is a key property of systems that are capable of mitigating or withstanding system faults, including gradual degradation.

A number of stochastic approaches to fault identification and mitigation have been developed in the past, with the main objective of estimating the \emph{remaining useful life} (RUL) of a system, and how this metric is influenced by the controller. Mo and Xie \cite{Mo2016} developed an approach to approximate the loss in effectiveness cause by actuator component degradation using a reliability value. Their method relies on frequency domain analysis using the Laplace transform, which is limited to linear systems; in turn, proposed reliability improvements hinge on the use of a PID controller strategy and rely on a particle swarm optimization routine, which is highly restrictive with regard to runtime constraints and convergence guarantees. A similar approach was developed by Si \emph{et al.} \cite{Si2020}, where reliability was assessed using an event-based Monte Carlo simulation approach, wherein potential degradation modes are simulated \emph{en masse}, further limiting the applicability of this method. This is due to the intractable number of potential failure modes that may be encountered in practice, which would demand a very large number of Monte Carlo simulations.

%Si \emph{et al.} \cite{Si2020} adopt a stochastic modeling approach, in which the state- and control-dependent evolution of a hidden degradation state is modeled to predict the remaining useful life of the actuator. The fault tolerant controller introduced in this work is once again a PID controller acting on a linear parameter-varying system that is adjusted by a nonlinear optimization routine. Much like in \cite{Mo2016}, the approach adopted in \cite{Si2020} is also dependent on Monte Carlo simulations to fit a stochastic diffusion process, imposing stringent restrictions on real-time capability and computational resources. 

In the deterministic setting, Wang \emph{et al.} \cite{Wang2014a} considered control input map degradation and actuator saturation in discrete-time linear systems, where a fault-tolerant control is developed by solving a constrained optimization problem. Given the discrete-time linear system setting, \cite{Wang2014a} uses efficient linear matrix inequality (LMI) techniques for controller synthesis. However, the class of actuator degradations considered in \cite{Wang2014a} is limited to linear diagonal control authority degradation with input saturation. In the context of switching systems, Niu \emph{et al.} \cite{Niu2022} considered the problem of active mode discrimination (AMD) with temporal logic-constrained switching, where a set of known switching modes was known \emph{a priori}. The AMD problem rests on a nonlinear optimization routine, which depends directly on temporal logic constraints and known switching modes that are often not known in advance.

In the present work, we consider a class of faults, which we refer to as \emph{actuator degradation}. The latter may arise as a result of wear and tear, software errors, or even adversarial intervention. Considering the following nonlinear control-affine dynamics
%\begin{equation*}
$
    \dot{x}(t) = f(x(t)) + g(x(t)) u(t),
$
%\end{equation*}
we define input degradation modes of the form
%\begin{equation*}
$
	\dot{\bar{x}}(t) = f(\bar{x}(t)) + R g(\bar{x}(t)) P u(t),
$
%\end{equation*}
where $P$ and $R$ are two unknown time-varying maps. We refer to $P$ as a \emph{control authority degradation map} (CDM), whereas $R$ is referred to as a \emph{control effectiveness degradation map} (CEM). 
Our focus in this work is on CDMs; a number of common CDMs are illustrated in Fig.~\ref{fig:cdm intro}.
A CDM $P$ effectively acts as a control input remapping, and can be thought of in the context of control systems with delegated control allocation, e.g., when an actuator with internal dynamics takes $u(t)$ and remaps it based on its internal state. Such a setting includes common degradation modes such as deadband or saturation, or any other nonlinear transformation due to effects such as friction. In more extreme cases, it is possible that $P$ maps a control signal $u_i (t)$ to another control signal $u_j (t)$ due to incorrect wiring or software design. The types of control authority degradation maps that we allow for in this work go beyond linear maps applied to discrete-time finite-dimensional linear systems, which hitherto been the main focus in prior work. We develop an \emph{efficient passive algorithm for detection and identification of CDMs}, with the quality of the reconstructed CDM monotonically increasing with system run time. Using this reconstruction of the CDM, we develop a \emph{novel method for viabilizing control signals}, with tight approximation error bounds that decrease with system run time.

We note that we do not consider external disturbances or other unmodeled dynamics in this work; robustness results regarding the effects of disturbances will be the subject of future work.
The results of this work allow for \emph{guaranteed approximation of arbitrary control degradation maps} without the need for knowledge of possible degradation modes or handcrafted filters, addressing an open problem in the literature The natural next step of this work, outside of the scope of this letter, is to approximate unviable control signal with their closest viable counterpart, with robustness bounds on the maximum trajectory deviation.

% We do emphasize the fact that the present approach allows for \emph{guaranteed approximation of arbitrary control degradation maps} without the need for knowledge of possible degradation modes or handcrafted filters, addressing an open problem in the literature.

\begin{figure}[t]
    \includegraphics[width=\linewidth]{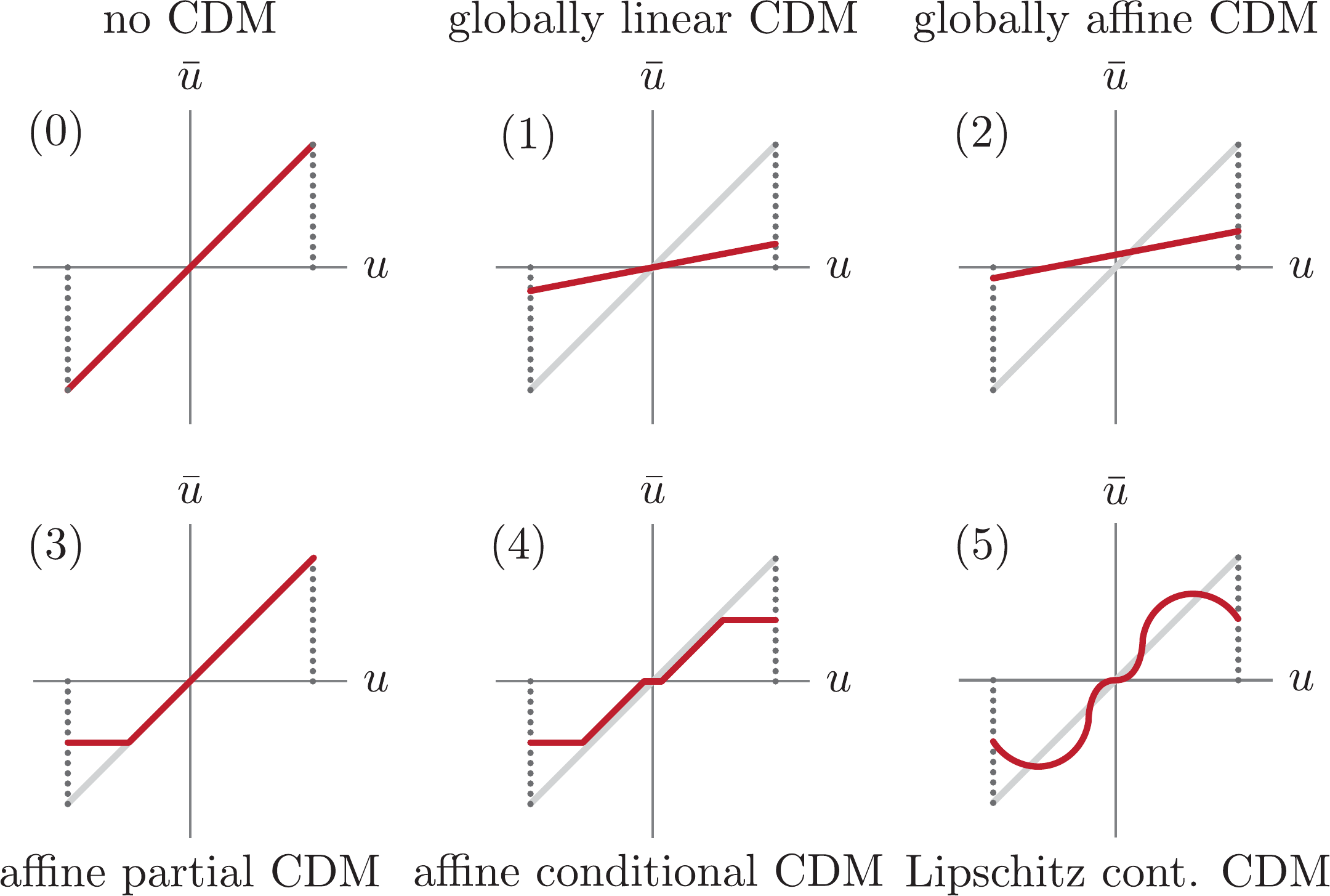}
    \caption{Comparison between various classes of control authority degradation maps.}
    \label{fig:cdm intro}
    \vspace{-0.5cm}
\end{figure}

%The paper has the following structure. We present preliminary concepts and notation, as well as the formal problem addressed in this work, in Sec.~\ref{sec:preliminaries}. The problem of CDM identifiability is addressed in Sec.~\ref{sec:identifiability} for multi-mode affine CDMs and Lipschitz continuous CDMs. 
%%Viability under CDMs is studied in Sec.~\ref{sec:viability}, where results on viabilizing maps are developed. 
%Sec.~\ref{sec:application} showcases the developed CDM identification and mitigation method on a control-oriented partial differential electrosurgical process model.
%%Conclusions are drawn in Sec.~\ref{sec:conclusion}.

\section{Preliminaries}\label{sec:preliminaries}

\subsection{Notation}

We use $\Vert \cdot \Vert$ to denote the Euclidean norm. Given two sets $A, B \subseteq \mathbb{R}^n$, we denote by $A + B$ their Minkowski sum $\{a + b : a \in A, b \in B\}$; the Minkowski difference is defined similarly. By $2^A$ we refer to the power set of $A$, i.e., the family of all subsets of $A$. We denote a closed ball centered around the origin with radius $r > 0$ as $\mathcal{B}_r$. By $\mathcal{B}(x, r)$ we denote $\{x\} + \mathcal{B}_r$. We denote by $\mathcal{L}(A, B)$ the set of bounded linear operators, and by $\mathcal{C}(A, B)$ the set of closed linear operators between $A$ and $B$. We define $\mathbb{R}_+ := [0, \infty)$.  For two points in a Banach space $\mathscr{B} \ni a, b$, let $[a, b]$ denote the convex hull of $a$ and $b$, i.e., $[a, b] := \mathrm{conv}\{a, b\}$. Given a point $x \in S$ and a set $A \subseteq S$, we denote $d(x, A) := \inf_{y \in A} d(x, y)$. We define the distance between two sets $A, B \subseteq \mathbb{R}^n$ to be
\begin{equation}\label{eq:distance between sets}
    d(A, B) := \sup_{a \in A} \inf_{b \in B} \Vert a - b \Vert.
\end{equation}
We denote the Hausdorff distance as
\begin{equation}\label{eq:Hausdorff distance}
    d_{\mathrm{H}} (A, B) := \max\{d(A, B), d(B, A)\},    
\end{equation}
An alternative characterization of the Hausdorff distance reads:
% \cite[pp.~280--281]{Munkres2000}:
\begin{equation}\label{eq:Hausdorff characterization}
    d_{\mathrm{H}} (A, B) = \inf \{ \rho \geq 0 : A \subseteq B_{+\rho}, B \subseteq A_{+\rho} \},
\end{equation}
where $X_{+\rho}$ denotes the \emph{$\rho$-fattening} of $X$, i.e., $X_{+\rho} := \bigcup_{x \in X} \{ y \in \mathbb{R}^n : \Vert x - y \Vert \leq \rho \}$.

We denote by $\partial A$ the boundary of $A$ in the topology induced by the Euclidean norm. For a function $g : A \to B$, we denote by $g^{-1}$ the inverse of this function if an inverse exists and otherwise denoting the preimage. By $\mathrm{dom}(g)$ we refer to the domain of the function (in this case $A$). We denote by $g^\dagger$ the Moore--Penrose pseudo-inverse of a linear function $g$. We use the Iverson bracket notation $\llbracket \cdot \rrbracket$, where the value is $1$ if the expression between the brackets is true, and $0$ otherwise.

In this work, we shall consider star-shaped sets, which are defined as follows:
% {\cite[\S15, p.~128]{Rockafellar1970}}
\begin{definition}[Star-shaped Set and MGFs]
    We call a closed compact set $K \subseteq \mathscr{B}$ \emph{star-shaped} if there exist (i) $\varsigma \in K$, and (ii) a unique function $\varrho : \mathcal{B}_1 \to \mathbb{R}_+$, such that:
    $
        K = \bigcup_{l \in \mathcal{B}_1} [\varsigma, \varsigma + \varrho(l) l]
    $
     where $\mathcal{B}_1$ denotes the unit ball in $\mathscr{B}$. We call $\varrho$ a \emph{Minkowski gauge function} (MGF), and $\varsigma$ the \emph{star center}.
\end{definition}

\subsection{Problem Formulation}

Consider a known nonlinear control-affine system of the form of
\begin{equation}\label{eq:nom sys}
    \dot{x}(t) = f(x(t)) + g(x(t)) u(t),
\end{equation}
where $x \in X$, $u \in U \subseteq \mathscr{U}$, $X$ and $\mathscr{U}$ are Hilbert spaces, and $f : X \to X$ and $g : X \to \mathcal{L}(\mathscr{U}, X)$. In this work, we assume $\mathscr{U} = \mathbb{R}^m$. In addition, we assume that $U$ is a \emph{star-shaped} subset of $\mathbb{R}^m$ such that $\mathrm{span} \ U = \mathbb{R}^m$. Finally, we assume that the full-state of the degraded system,
\begin{equation}\label{eq:degraded sys}
	\dot{\bar{x}}(t) = f(\bar{x}(t)) + R g(\bar{x}(t)) P u(t),
\end{equation}
is known without error.

In system \eqref{eq:degraded sys}, a control action degradation map $R$ can model changes in the control allocation function $g$, which may include actuator reconfiguration, such as a change in the trim angle on aircraft control surfaces, or misalignment of actuators due to manufacturing imperfections or wear and tear. Since $R$ acts after $g$, it does not directly remap the control signal $u(t)$, but it changes the action of a control input on the system; we therefore talk about control effectiveness, as opposed to control authority in the case of $P$, which acts before $g$. Changes in the drift dynamics $f(x(t))$ will not be treated in this work.

In addition to identifying or approximating CDM $P$, we are interested in `undoing' the effects of control authority degradation as much as possible. In particular, we are interested in the set of control signals \eqref{eq:nom sys} that can still be replicated in \eqref{eq:degraded sys} when the CDM is acting; we call this the set of \emph{viable control inputs}, $U_{\mathrm{v}}$. With knowledge of $P$, we  develop in this work a method to obtain, for $u_{\mathrm{cmd}} \in U_{\mathrm{v}}$, $u_{\mathrm{v}}$ such that $P u_{\mathrm{v}} = u_{\mathrm{cmd}}$; here, $u_{\mathrm{cmd}}$ and $u_{\mathrm{v}}$ are called commanded and viabilized control inputs, respectively. This approach is closely related to a technique known in the literature as \emph{fault hiding} \cite{Richter2011}. Fault hiding is achieved by introducing an output observer based on the output of the degraded system, and augmenting the nominal system model by introducing so-called virtual actuators, which requires a nonlinear reconfiguration block that is strongly dependent on the underlying problem structure and failure modes \cite[\S3.6, p.~42]{Richter2011}. In the setting considered in this work, we show that we can adopt the fault hiding philosophy under much less stringent constraints for a general class of systems and degradation modes.

%A moving horizon extension to the fault hiding method was presented by Costa \emph{et al.} \cite{Costa2021} based on model predictive control approach using a modified cost function. This approach, while adaptive in nature, has limited application to linear discrete-time systems. The approach that we present in this work is capable of online reconfiguration to unknown control authority degradation modes, and applies to a large class of control-affine nonlinear systems \eqref{eq:nom sys} with fully observable state. The resulting viabilizing remapping is similar to the fault hiding technique, but does not require virtual actuators or the solution of an optimization problem.

%Since the set of inviable controls $U \setminus U_{\mathrm{v}}$ is often nonempty, we will also have to deal with inviable commanded control inputs; in this case, we study the error introduced by taking the closest input $u_{v, \mathrm{cmd}} = \mathrm{proj}_{U_{\mathrm{v}}} u_{\mathrm{cmd}}$. Such error bounds go beyond approaches found in the literature as far as the authors are aware. Ultimately, we are interested in guaranteed \emph{closed-loop viability} \cite[Ch.~5, p.~213]{Aubin1984}, which in this work refers to the capability of the system to exhibit expected behavior despite actuator degradation.

In this work, we are interested in modeling unknown degraded system dynamics \eqref{eq:degraded sys} for a time-invariant control authority degradation map (CDM) $P : U \to \bar{U}$, and no control effectiveness degradation (i.e., $R = I$). This amounts to reconstructing, or \emph{identifying}, $P$:

\begin{problem}[Identifiability of Control Authority Degradation Maps]\label{prob:identifiability}
	For a class of time-invariant CDMs $\mathcal{P} \in P$,	if possible, identify $P$ based on a finite number of full state, velocity, and control input observations ($\bar{x}(t)$, $\dot{\bar{x}}(t)$, $u(t)$) of the degraded system.
%	, given mild assumptions on $g(x)$? If so, is there a constructive method to obtain $P$?
\end{problem}

Ideally, we would like to identify general nonlinear CDMs with known bounds on the approximation error. 
%In this work, we deal with the following classes of CDMs:
%\begin{enumerate*}[label=(\arabic*)]
%	\item globally acting linear CDMs;
%	\item globally acting affine CDMs;
%	\item partial CDMs (p-CDMs), where the affine degradation mode acts on a subset of $U$;
%	\item conditional CDMs (c-CDMs), which generalize p-CDMs by including multiple affine degradation modes acting on disjoint subsets of $U$;
%	\item Lipschitz continuous (nonlinear) c-CDMs as approximated by $N$-mode affine c-CDMs with explicit error bounds.
%\end{enumerate*}
We illustrate the control authority degradation modes that are covered in this work in Fig.~\ref{fig:cdm intro}.

We now proceed by solving Problem~\ref{prob:identifiability} for an unknown multi-mode affine CDMs, which allows for approximating Lipschitz continuous nonlinear CDMs with bounded error.

\section{Identifiability of Control Authority Degradation Maps}\label{sec:identifiability}

We now consider Problem~\ref{prob:identifiability}. Let us assume that for $U$, the Minkowski gauge function $\varrho$ is known. Let $P : U \to \bar{U}$ be an unknown control authority degradation map (CDM). We assume that $\bar{U}$ is also a star-shaped set, providing conditions on $P$ and $U$ under which this holds. It bears mentioning that star-shaped sets are more general than convex sets; most results presented in this work will apply to star-shaped sets, which include polytopes, polynomial zonotopes, and ellipsoids. 
%Throughout this work, we also assume full-state observability of the closed-loop system; this assumption will also be relaxed further on in this work.

Before we provide any results on the identifiability of control authority degradation modes, we pose the following key assumption on the nominal system dynamics \eqref{eq:nom sys}. We allow for an \emph{infinite-dimensional} state-space $X$, that is to say, $X$ is a set of \emph{functions}, but $X = \mathbb{R}^n$ is also captured:

\begin{assump}\label{assump:inf-dim affine control degradation}
For system \eqref{eq:degraded sys}, assume that 
\begin{enumerate}[label=\roman*.]
    \item $g(x)$ has closed range for all $x \in X$;
    \item $g(x)$ is injective for all $x \in X$, i.e., $\mathrm{ker}(g(x)) = \{0\}$;
    \item $\dot{x}$ is known at some $x \in X$ with $u = 0$.
\end{enumerate}
\end{assump}

\begin{remark}\label{remark:fin-dim affine control degradation}
    In the case of finite-dimensional systems, i.e., $X \subseteq \mathbb{R}^n$, the first two conditions of Assumption~\ref{assump:inf-dim affine control degradation} can be stated as:
    \begin{enumerate}[label=\roman*.]
        \item The system is not overactuated, i.e., $m \leq n$;
        \item $g(x)$ is of full-column rank for all $x \in X$.
    \end{enumerate}
\end{remark}

We shall consider the case of multiple control degradation modes acting throughout the space $U$. The simplest of the so-called \emph{conditional control authority degradation modes} (c-CDMs) acts only on a compact subset of $U$; we refer to these c-CDMs as \emph{partial control authority degradation modes} (p-CDMs). Consider two compact star-shaped sets $\check{U}, \hat{U} \subseteq U$, and two p-CDMs
\begin{equation}
    P_{\check{U}}(u) := u + \llbracket u \in \check{U} \rrbracket (P - I) u,
\end{equation}
\begin{equation}
    P_{\hat{U}}(u) := u + \llbracket u \not\in \hat{U} \rrbracket (P - I) u,
\end{equation}
for some control degradation map $P$. Here, $P_{\check{U}}$ is an \emph{internally acting} partial CDM (i.e., acting inside $\check{U}$), whereas $P_{\hat{U}}$ is an \emph{externally acting} partial CDM (acting outside $\hat{U}$); when this distinction is immaterial, we use a combined hat and check symbol (e.g., $\hatcheck{U}$), where $\hatcheck{U}$ is simply called the \emph{affected set} of control inputs.

In reconstructing an $N$-mode c-CDM, we face the problem of discerning which control inputs belong to which conditional degradation mode. To make this problem tractable, we pose the following assumption:

\begin{assump}\label{assump:N-mode c-CDM}
Let the internally acting $N$-mode c-CDM satisfy	 the following properties:
\begin{enumerate}[label=\roman*.]
	\item The number of modes $N$ is known;
	\item $\check{\mathcal{U}}$ is a family of convex sets;
	\item $\check{\mathcal{P}}$ is a family of affine maps denoted by $Q_i = p_i + P_i$.
	\item There exists a known $\delta > 0$, such that for all $i \neq j$, $d_{\mathrm{H}} \left( (\check{U}_i, \check{P}_i \check{U}_i), (\check{U}_j, \check{P}_j \check{U}_j) \right) \geq \delta$.
\end{enumerate}
\end{assump}

We are also interested in obtaining outer-approximations of $\check{U}$ and inner-approximations of $\hat{U}$ for each degradation mode, as illustrated in Fig.~\ref{fig:mgf comp}, so that we can restrict control inputs to regions that are guaranteed to be unaffected. Since we only have access to a finite number of control input samples, we pose the following assumption regarding the regularity of the MGF associated with $P \hatcheck{U}$.

\begin{assump}\label{assump:lip MGF}
	Assume that $\hatcheck{U}$ has star center $\hatcheck{\varsigma} = 0$, and assume that the MGF $\hatcheck{\varrho}$ associated with $\hatcheck{U}$ is Lipschitz continuous, i.e., there exists a known $\hatcheck{L}$ such that
	$
		|\hatcheck{\varrho}(l) - \hatcheck{\varrho}(l')| \leq \hatcheck{L} \Vert l - l' \Vert,
	$
	for all $l, l' \in \mathcal{B}_1$.
\end{assump}

We now proceed to show that Assumption~\ref{assump:lip MGF} holds for the image of Lipschitz star-shaped sets under affine maps.

\begin{lemma}
	Given a star-shaped set $U$ characterized by a Lipschitz MGF $\varrho$ and star center $\varsigma$, the range of $U$ under an affine map $Q u := p + P u$ is also a star-shaped set with Lipschitz MGF.
\end{lemma}

%\begin{proof}
%	A straightforward expansion shows that $Q(U)$ has $\hatcheck{\varsigma} = p + P \varsigma$, and a MGF $\hatcheck{\varrho}(l) = \varrho(l) P l$. If $\varrho$ has Lipschitz constant $L$, then a valid Lipschitz constant for $\hatcheck{\varrho}$, can be found by the boundedness of $\varrho$ and finite operator norm of $P$ as follows:
%	\begin{equation*}
%	\begin{split}
%		\Vert \hatcheck{\varrho}(l) - \hatcheck{\varrho}(l') \Vert &= \Vert \varrho(l) P l - \varrho(l') P l' \Vert \\
%		&= \Vert (\varrho(l) - \varrho(l')) P l + \varrho(l') P (l - l') \Vert \\
%%		&\leq \Vert \varrho(l) - \varrho(l') \Vert \Vert P l \Vert + \Vert P (l - l') \Vert \Vert \varrho(l') \Vert \\
%		&\leq (\Vert P \Vert L + \Vert P \Vert \Vert \varrho \Vert_\infty) \Vert l - l' \Vert.
%	\end{split}
%	\end{equation*}
%\end{proof}

We can now pose a key result on the guaranteed approximation of Lipschitz MGFs from a finite set of samples.

\begin{proposition}\label{prop:approx MGF}
	Assume that Assumption~\ref{assump:lip MGF} holds for the unknown MGFs $\check{\varrho}$ and $\hat{\varrho}$. Then, for some given $\check{u}, \check{u}' \in \check{U}$ and $\hat{u}, \hat{u}' \not\in \hat{U}$, we have for all $\mu \in [0, 1]$:
	\begin{equation}\label{eq:MGF outer}
	\begin{split}
		&\check{\varrho} \left( \frac{\mu \check{l} + (1-\mu) \check{l}'}{\Vert \mu \check{l} + (1-\mu) \check{l}' \Vert} \right) \leq \\
		&\min \left\{\Vert \check{u} \Vert + (1-\mu) \check{L} \Vert l - l' \Vert, \Vert \check{u}' \Vert + \mu \check{L} \Vert \check{l} - \check{l}' \Vert \right\},
	\end{split}
	\end{equation}
	and
	\begin{equation}\label{eq:MGF inner}
	\begin{split}
		&\hat{\varrho} \left( \frac{\mu \hat{l} + (1-\mu) \hat{l}'}{\Vert \mu \hat{l} + (1-\mu) \hat{l}' \Vert} \right) \geq \\
		&\max \left\{ 0, \Vert \hat{u} \Vert - (1-\mu) \hat{L} \Vert \hat{l} - \hat{l}' \Vert, \Vert \hat{u}' \Vert - \mu \hat{L} \Vert \hat{l} - \hat{l}' \Vert \right\},
	\end{split}
	\end{equation}
	where $\hatcheck{l} := \hatcheck{u}/\Vert \hatcheck{u} \Vert$ and $\hatcheck{l}' := \hatcheck{u}'/\Vert \hatcheck{u}' \Vert$.
\end{proposition}

\begin{proof}
	This result follows directly from non-negativity of the MGF and the mean value theorem, given the Lipschitz continuity of $\hatcheck{\varrho}$ as assumed in Assumption~\ref{assump:lip MGF}.
\end{proof}

%\begin{remark}\label{rem:star center not unique}
%%	In the case of $\hatcheck{\varsigma} \neq 0$, we may assume convexity and take $\hatcheck{\varsigma}$ to be the mean of all given $\hatcheck{u}$'s; in any other case, more structure of $\hatcheck{U}$ is required to be able to reconstruct $\hatcheck{\varsigma}$ from a set of $\hatcheck{u}$'s.	
%	We assume that the star center is at the origin ($\hatcheck{\varsigma} = 0$), since it is impossible to uniquely determine the location of the star center given a finite set of points in general. As a matter of fact, the set of star centers of a star-shaped set, also known as its \emph{kernel}, is not generally a singleton, and hence star centers are not unique in general \cite[p.~1005]{Hansen2020}. In the case of a convex set, the kernel of a star-shaped set is equal to the set itself, making the choice of the star-shaped set trivial. In the case of a convex set $\hatcheck{U}$, it is sufficient to require that $0 \in \hatcheck{U}$ to apply Proposition~\ref{prop:approx MGF} as is; alternatively, we can take any point in $\hatcheck{U}$ to be the star-center and apply \eqref{eq:MGF outer}--\eqref{eq:MGF inner} with respect to this new star center, which need not be the origin. This procedure yields a new MGF with respect to this new star center.
%\end{remark}

The results given in Proposition~\ref{prop:approx MGF} allow for direct inner-approximation of $\hat{U}$ and outer-approximation of $\check{U}$; these results will allow us to restrict closed-loop control inputs to a subset of $U$ that is \emph{guaranteed to be unaffected by $P$} as illustrated in provided in Fig.~\ref{fig:mgf comp}. The method for approximating $\hatcheck{U}$ will be rigorized in the next theorem.

\begin{figure}[t]
	\centering
    \includegraphics[width=0.75\linewidth]{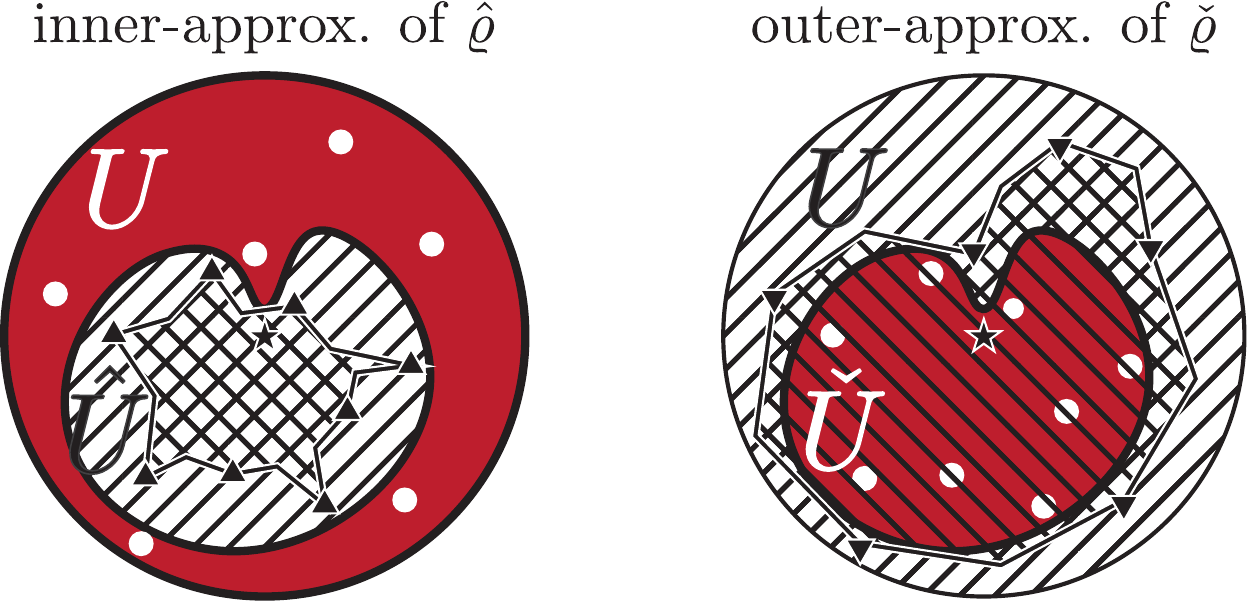}
    \caption{Comparison between inner- and outer-approximations of $\hat{U}$ and $\check{U}$ respectively, based on Proposition~\ref{prop:approx MGF} and Theorem~\ref{thm:N-mode fin-dim aff c-CDM} for a 1-mode c-CDM. The region with top-right-pointing hatching indicates the set in which the control input is unaffected; the red-colored region indicates the affected set. The respective approximations of $\protect\hatcheck{U}$ allow one to find regions in which control inputs are guaranteed to be unaffected. In the left image, the set indicated by top-left-pointing hatching is an inner-approximation of $\hat{U}$, and in the right image this set is an outer-approximation of $\check{U}$.}
    \label{fig:mgf comp}
    \vspace{-0.5cm}
\end{figure}

We now pose the main result on the identifiability of $N$-mode conditional control authority degradation modes (c-CDMs), where multiple affine CDMs act on disjoint subsets of $U$; this will allow us to approximate of Lipschitz continuous CDMs as shown at the end of the next section.

\begin{theorem}[Reconstructing $N$-mode Affine c-CDMs]\label{thm:N-mode fin-dim aff c-CDM}
    Consider system \eqref{eq:degraded sys} and Assumptions~\ref{assump:inf-dim affine control degradation}--\ref{assump:N-mode c-CDM}. Assume that the c-CDM is represented by $N$ unknown internally acting \emph{affine} maps $Q_i$, each acting on mutually disjoint unknown star-shaped sets $\check{U}_i \subseteq U$, giving $Q_{\check{\mathcal{U}}}$ as the p-CDM. Let there be a given array of distinct state--input pairs $[(\bar{x}[i], u[i])]_{i=1}^N$, and a corresponding array of degraded velocities $[\dot{\bar{x}}[i]]_{i=1}^{N'}$ obtained from system \eqref{eq:degraded sys}, with $N' \geq N(m+1)$. Let there also be a given array of undegraded state--input pairs $[(x_*[i], u_*[i])]_{i=1}^M$, with $M \geq m$. Assume that there exist $m$ state--input pairs indexed by $\jmath$ and $\jmath_*$, such that the arrays of input vectors $\{u[\jmath_j]\}_{j=1}^m$ and $[u_*[\jmath_{*, j}]]_{j=1}^m$ are linearly independent. 
    
    Cluster the array $[(u[i], \hatcheck{u}[i])]_{i=1}^{N'}$ into $N$ clusters with a Hausdorff distance of at least $\delta$ between each pair of clusters. If each cluster $i$ contains at least $m$ vectors $u[i]$ that are linearly independent, then
%    , then Theorem~\ref{thm:fin-dim aff p-CDM} provides a method to uniquely obtain $Q_i$ and an outer-/inner-approximation of $\check{U}_i$. 
	$Q_{\check{\mathcal{U}}}$ can be approximated as follows:
    \begin{equation}\label{eq:guaranteed approx of N-mode c-CDM}
    	\tilde{Q}_{\check{\mathcal{U}}} u = \begin{cases}
    		u & u \not\in \check{U}_{\mathrm{outer}}, \\
    		\sum_{i=1}^N \llbracket u \in \check{U}_{i, \mathrm{inner}} \rrbracket Q_i u \quad & u \in \check{U}_{\mathrm{inner}}, \\
    		\mathrm{inconclusive} & u \in \check{U}_{\mathrm{outer}} \setminus \check{U}_{\mathrm{inner}},
    	\end{cases}
    \end{equation}
    where $\check{U}_{\mathrm{inner}} := \bigcup_{i=1}^N \check{U}_{i, \mathrm{inner}}$ and $\check{U}_{\mathrm{outer}} := \bigcup_{i=1}^N \check{U}_{i, \mathrm{outer}}$. Each $Q_i$ is obtained by considering for each cluster $v_i := g^\dagger(\bar{x}[j])(\dot{\bar{x}}[j] - f(\bar{x}[j]))$ where index $j$ is not part of the array of linearly independent inputs indexed by $\jmath$, $\vec{u} := [\begin{smallmatrix} u[\jmath_1] - v_i & \cdots & u[\jmath_m] - v_i \end{smallmatrix}]$ and
    $
        \Delta \vec{u} := \left[ g^\dagger(\bar{x}[\jmath_j]) (\dot{\bar{x}}[\jmath_j] - f(\bar{x}[\jmath_j])) - u[\jmath_j] - v_i \right]_{j=1}^m.
    $
    Linear operator $P_i$ is obtained as
    \begin{equation}\label{eq:lin CDM}
        P_i = (\vec{u} + \Delta \vec{u}) \vec{u}\transp (\vec{u} \vec{u}\transp)\inv.
    \end{equation}
    The translation $p_i$ is obtained as $p_i = v_j - P_i u[j]$, which yields the $i$'th mode affine CDM $Q_i$:
    \begin{equation}\label{eq:aff CDM}
		Q_i u := p_i + P_i u.
	\end{equation}
	Here, each affected set is approximated as follows: In case $\hatcheck{U}_i$ is internally acting (i.e., $\check{U}_i = \hatcheck{U}_i$), \eqref{eq:MGF outer} yields an outer-approximation to $\check{\varrho}_i$ by taking a convex combination of the $m$ basis vectors $\{\hatcheck{l}_*[\jmath_{*, j}] = u_*[\jmath_{*, j}]/\Vert u_*[\jmath_{*, j}] \Vert \}_{j=1}^m$ and their values. Similarly, for externally acting $\hatcheck{U}_i$ (i.e., $\hat{U}_i = \hatcheck{U}_i$), \eqref{eq:MGF inner} yields an inner-approximation to $\hat{\varrho}_i$ using $m$ basis vectors $\{\hatcheck{l}[\jmath_{j}]\}_{j=1}^m$. Inner- and outer-approximations satisfy the relation $\hatcheck{U}_{i, \mathrm{inner}} \subseteq \hatcheck{U}_i \subseteq \hatcheck{U}_{i, \mathrm{outer}}$ (cf. Fig.~\ref{fig:mgf comp}).
\end{theorem}

\begin{proof}
	We first consider a globally acting affine CDM. We obtain the closed-form expression of $P_i$, \eqref{eq:lin CDM}, by solving the quadratic program $\min_{P \in \mathcal{L}(U, \bar{U})} \Vert P \vec{u} - (\vec{u} + \Delta \vec{u}) \Vert^2$,
    which yields a unique linear map $P$ that maps $\vec{u}$ to $\vec{u} + \Delta \vec{u}$ as desired. The translation term $p_i$ can be verified by direct substitution in \eqref{eq:aff CDM}, yielding the affine map $Q_i$.
    
    In \eqref{eq:lin CDM}, since the inverse of $\vec{u}\transp \vec{u}$ must be taken, we require both that $\vec{u}$ is a square matrix, and $\vec{u}\transp \vec{u}$ is invertible. This is achieved by considering $\vec{u} \in \mathbb{R}^{m \times m}$ of full column rank, as guaranteed by the linear independence hypothesis.
    
    Regarding $g^\dagger(x)$, the Moore--Penrose pseudo-inverse is defined for a general Hilbert space $X$, provided that $\mathrm{range}(g(x))$ is closed for all $x \in X$ \cite[\S4.2, p.~47]{Nashed1976}. For $g^\dagger(x)$ to be a left-inverse, a necessary condition is that $g(x)$ be injective, i.e., $\mathrm{ker}(g(x)) = \{0\}$ for all $x \in X$ \cite[Cor.~2.13, p.~36]{Nashed1976}. Finally, the translation term $p$ is accounted for as well \eqref{eq:aff CDM}.
	
	To approximate the $i$'th affected set, $\hatcheck{U}_i$, we require a spanning set of basis vectors that lie within $\hatcheck{U}_i$, as provided for in the hypotheses. The unknown MGF associated with $\hatcheck{U}_i$ can be obtained according to Proposition~\ref{prop:approx MGF} using \eqref{eq:MGF outer}--\eqref{eq:MGF inner}, where an inner-approximation is desired for internally acting p-CDMs, and outer-approximations for externally acting p-CDMs. These approximations are obtained through repeated convex combinations and the corresponding inequality given in \eqref{eq:MGF outer}--\eqref{eq:MGF inner}, for a total of $m$ times; an explicit expansion of the resulting expression is omitted here for the sake of space.
\end{proof}

\begin{remark}
	This result incorporates p-CDMs that map a set $\check{U}$ to a constant, e.g., $Q_{\check{\mathcal{U}}} \check{U} = p$. To highlight the utility of this result, it should be noted that the \emph{hypotheses given here allow for commonly encountered degradation modes such as deadzones and saturation to be modeled} (see Fig.~\ref{fig:cdm intro}(4)). Additionally, Theorem~\ref{thm:N-mode fin-dim aff c-CDM} allows for \emph{discontinuous} control authority degradation modes, a property that is rarely present in prior work.
\end{remark}

We can now consider the case in which $P$ is a Lipschitz continuous CDM. We consider an approximation of $P$ by an $N$-mode affine c-CDM $\tilde{P}$, for which we derive an explicit error bound given that the Lipschitz constant of $P$, $L_P$, is known.

\begin{theorem}[Approximating Lipschitz continuous CDMs by $N$-mode Affine c-CDMs]\label{thm:Lip nonlin CDM}
	Let the hypotheses of Theorem~\ref{thm:N-mode fin-dim aff c-CDM} hold, with the exception that $P := Q_{\check{\mathcal{U}}}$ is now a Lipschitz continuous CDM with Lipschitz constant $L_P$ and Assumption~\ref{assump:N-mode c-CDM} is now dropped. If $N$ clusters that satisfy the linear independence requirements of Theorem~\ref{thm:N-mode fin-dim aff c-CDM} are identified, then the resulting $N$-mode affine c-CDM approximation $\tilde{P}$ has the following error:
	
	For all $u \in \check{U}_{i, \mathrm{inner}}$ and all $i=1, \ldots, N$,
	\begin{equation}\label{eq:nonlinear CDM error bound}
		\Vert P u - \tilde{P} u \Vert \leq \Vert \min_{j=1,\ldots,m} \varepsilon_{i, j} + L_P \Vert u[i,j] - u \Vert,
	\end{equation}
	where $\varepsilon_{i, j} := \Vert P \vec{u}_i[j] - \check{P}_i \vec{u}_i[j] \Vert$, and $u[i, j] := \vec{u}_i [j]$, where $\vec{u}_i$ is an array composed of all control inputs in the $i$'th cluster.
\end{theorem}

\begin{proof}
	The proof is similar to that of Theorem~\ref{thm:N-mode fin-dim aff c-CDM}, with the error bound \eqref{eq:nonlinear CDM error bound} following an application of the triangle inequality in combination with the Lipschitz continuity of $P$, the properties of the affine maps $\check{P}_i$, and the known samples of $(u, P u)$.
\end{proof}

We can now pose a convergence result on the $N$-mode affine c-CDM approximation $\tilde{P}$ of a Lipschitz continuous CDM $P$.

\begin{corollary}\label{cor:Lip nonlin CDM convergence}
	Error bound \eqref{eq:nonlinear CDM error bound} is monotonically decreasing in the the number of samples $N'$ and the number of c-CDM modes $N$. In the limit of the $N', N \to \infty$, error bound \eqref{eq:nonlinear CDM error bound} converges to zero.
\end{corollary}

\begin{proof}
	In \eqref{eq:nonlinear CDM error bound}, $\varepsilon_{i,j}$ monotonically converges to zero, because the operator norm $\Vert P - \check{P}_i \Vert$ restricted to the $i$'th cluster converges monotonically to zero; this fact follows by considering that the diameter of each cluster converges to zero for a greater number of samples and clusters, similarly to the proof of Lemma~\ref{lm:convergence of approximations of U}, as well as the fact that $P$ is Lipschitz continuous, meaning that the total variation of $P$ on this restriction decreases monotonically as well. Another consequence of the diminishing cluster diameter is that $\Vert u[i,j] - u \Vert$ converges monotonically to zero.
\end{proof}

In the results given above, we find that it is in general impossible to uniquely determine each $\hatcheck{U}$ from finitely many samples. Intuitively, given a greater number of distinct points inside $\hatcheck{U}$ and $U \setminus \hatcheck{U}$, it should be possible to more tightly approximate $\hatcheck{U}$. This idea is illustrated in Fig.~\ref{fig:mgf approx comp}. We now state a lemma on the convergence of inner- and outer-approximations of the affected set $\hatcheck{U}$.

\begin{lemma}\label{lm:convergence of approximations of U}
	Consider $\epsilon > 0$, such that a given set of $N_\epsilon \geq m$ distinct pairs $(u, P_{\hatcheck{U}} u)$ denoted by $\mathcal{G} U_{N, \epsilon}$, satisfies Assumptions~\ref{assump:inf-dim affine control degradation}--\ref{assump:lip MGF}, where $P_{\hatcheck{U}}$ is (i) an $N$-mode affine c-CDM, or (ii) a Lipschitz continuous CDM. Let $\mathcal{G} U_N$ be such that for each $u_i$ in $\mathcal{G} U_{N, \epsilon}$, $\bigcup_{i=1}^{N_\epsilon} \mathcal{B}_{\epsilon} (u_i) \supseteq \hatcheck{U}$; i.e., $\epsilon$-balls centered at each sampled control input form a cover of $\hatcheck{U}$. Let $\hatcheck{U}^{N_\epsilon}_{\mathrm{inner}}$ and $\hatcheck{U}^{N_\epsilon}_{\mathrm{outer}}$ denote the corresponding inner- and outer-approximations of $\hatcheck{U}$ using the procedure given in Theorem~\ref{thm:N-mode fin-dim aff c-CDM} from $\mathcal{G} U_{N, \epsilon}$. Then, we have $\hatcheck{U}^{N_{\epsilon}}_{\mathrm{inner}} \subseteq \hatcheck{U}^{N_{\epsilon'}}_{\mathrm{inner}} \subseteq \hatcheck{U}$ and $\hatcheck{U} \subseteq \hatcheck{U}^{N_{\epsilon'}}_{\mathrm{outer}} \subseteq \hatcheck{U}^{N_{\epsilon}}_{\mathrm{outer}}$ for all $\epsilon' < \epsilon$. In addition, we have
	$
		\lim_{\epsilon \to 0} \hatcheck{U}^{N_\epsilon}_{\mathrm{inner}} = \lim_{\epsilon \to 0} \hatcheck{U}^{N_\epsilon}_{\mathrm{outer}} = \hatcheck{U}.
	$
\end{lemma}

\begin{proof}
	Since it is assumed that the pairs in $\mathcal{G} U_{N, \epsilon}$ are distinct, the approximations of $\check{\varrho}$ and $\hat{\varrho}$ obtained in Theorem~\ref{thm:N-mode fin-dim aff c-CDM} will become increasingly tight for decreasing $\epsilon$, since the expressions derived in Theorem~\ref{prop:approx MGF} will rely increasingly less on the Lipschitz bound assumption. Since $d_{\mathrm{H}} (\hatcheck{U}^{N_\epsilon}_{\mathrm{inner}}, \hatcheck{U}^{N_\epsilon}_{\mathrm{outer}})$ is monotonically decreasing for decreasing $\epsilon$, in the limit of $\epsilon \to 0$, both sequences will converge to $\hatcheck{U}$ in the Hausdorff distance. This follows from the fact that the Hausdorff distance between the boundary of $\hatcheck{U}$ and the sampled points $u$ decreases monotonically with decreasing $\epsilon$, leading to tighter approximations of $\hat{\varrho}$ and $\check{\varrho}$ as per Proposition~\ref{prop:approx MGF}.
\end{proof}

\begin{figure}[t]
	\centering
    \includegraphics[width=0.97\linewidth]{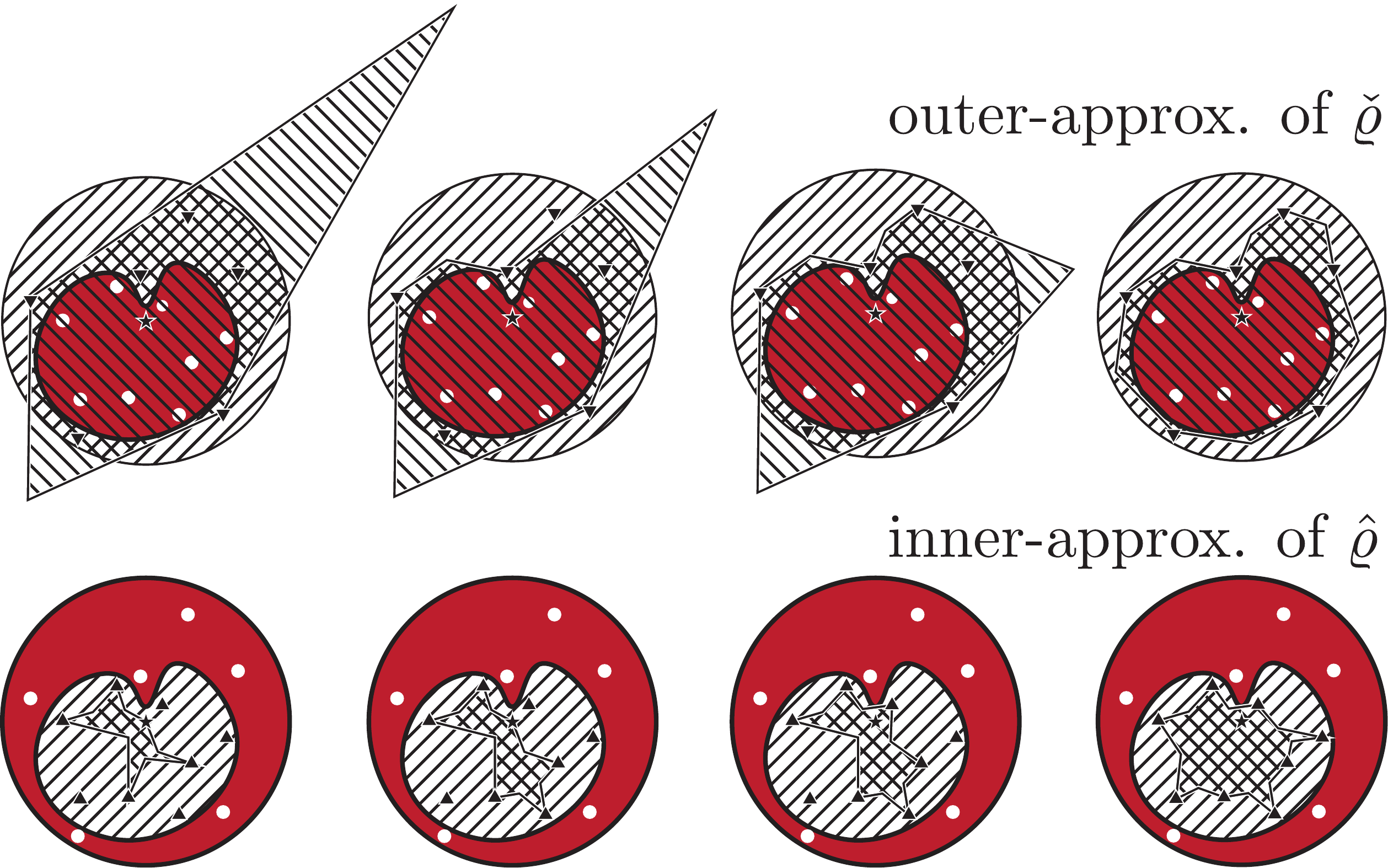}
    \caption{Comparison between inner- and outer-approximations of $\hat{U}$ and $\check{U}$ respectively, based on Proposition~\ref{prop:approx MGF} and Theorem~\ref{thm:N-mode fin-dim aff c-CDM} for an increasing number of samples for a 1-mode c-CDM. Clearly, for a larger number of points of sufficiently dispersed points, increasingly tight approximations are obtained as formalized in Lemma~\ref{lm:convergence of approximations of U}.}
    \label{fig:mgf approx comp}
    \vspace{-0.5cm}
\end{figure}

\begin{remark}
	In Lemma~\ref{lm:convergence of approximations of U}, note that the $\epsilon$-covering argument is required to ensure that the distinct points are sufficiently dispersed; simply considering $N \to \infty$ does not ensure convergence of the Hausdorff distance between the inner- and outer-approximation to zero. This fact can also be observed when looking at Fig.~\ref{fig:mgf approx comp}.
\end{remark}

\section{Application}\label{sec:application}

%We consider an infinite-dimensional system based on a 3D model of tissue thermodynamics during electrosurgery \cite{El-Kebir2022d}:
%\begin{equation}
%	\dot{x}(t, \xi) = a \nabla^2 x(t, \xi) + q(\xi - u_2) u_1,
%	\dot{l} &= 
%\end{equation}
%where $f(x) = a \nabla^2 x$ and $g(x) u = q(\cdot - u_2) u_1$, $u \in [0, 10] \times [0, 1]$. The unit heat source is modeled as $q(\xi) = \frac{1}{2\epsilon} \llbracket \xi \in [-\epsilon, \epsilon] \rrbracket$, for some known $\epsilon > 0$. We are interested in a 3-mode piecewise linear CDM, with $\check{U}_1 = [0, 5] \times [0, 0.25]$, $\check{U}_2 = [0, 5] \times [0.75, 1]$, and $\check{U}_3 = [5, 10] \times [0, 1]$. We take $\check{P}_1 u = (2.5 + 0.5 u_1, 0.125 + 0.5 u_2)$, $\check{P}_1 u = (9 - 0.2 u_1, u_2)$

We consider an infinite-dimensional system based on a 3D model of tissue thermodynamics during electrosurgery \cite{El-Kebir2023c}:
\begin{equation}\label{eq:pde sys}
\begin{split}
	\dot{z}(t, \xi) &= a \nabla^2 z(t, \xi) + q(\xi) u_1, \\
	\dot{d}(t) &= z(t, 1) u_2
\end{split}
\end{equation}
where $u \in [0, 10] \times [0, 1]$. The unit heat source is modeled as $q(\xi) = \frac{1}{\epsilon} \llbracket \xi \in [0, \epsilon] \rrbracket$, for some known $\epsilon > 0$. This model approximates a slab of tissue with the state representing the surface temperature; $u_1$ denotes the input power and $u_2$ denotes the needle depth.

\begin{figure}[t]
	\centering
    \includegraphics[width=0.61\linewidth]{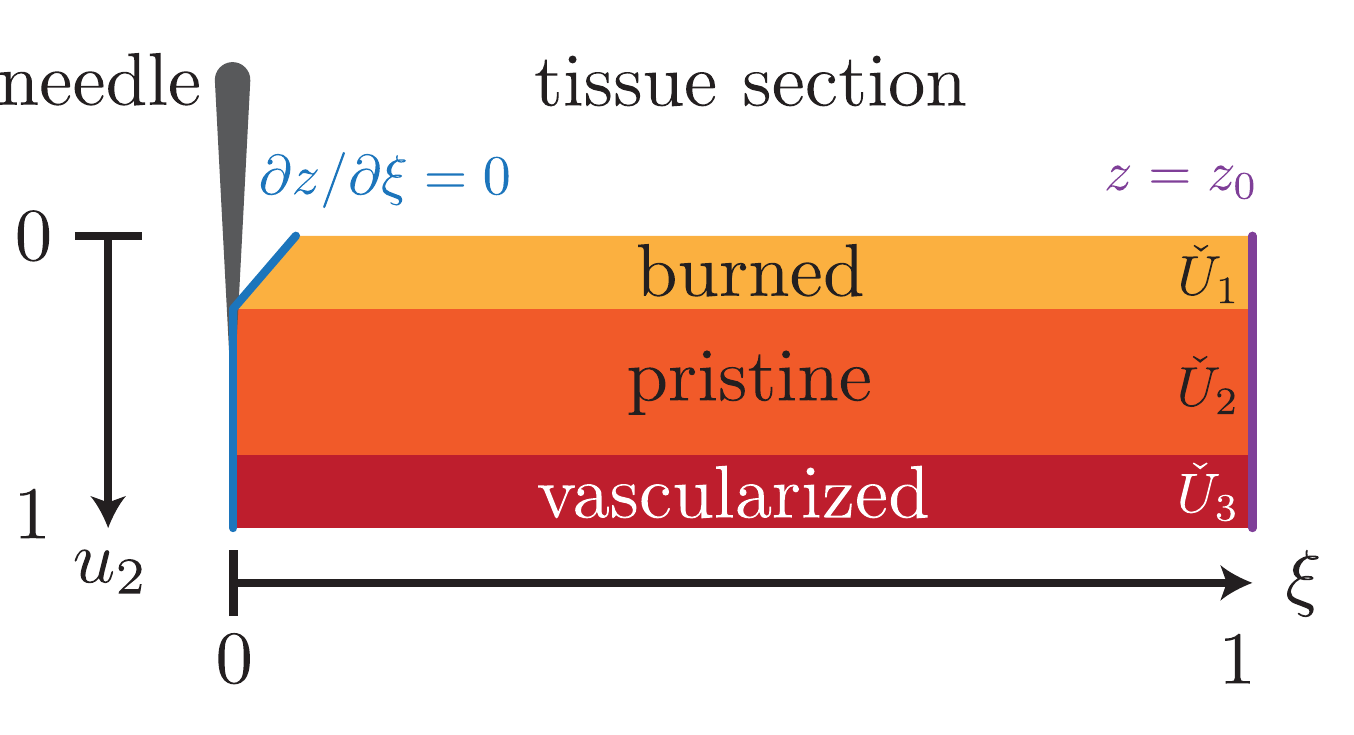}
    \caption{Illustration of the domain on which \eqref{eq:pde sys} acts. There are three regions of degradation, of which the burned and the vascularized region form two; the last region is not shown, but saturates the $u_1$ to 5.}
    \label{fig:pde-ex}
    \vspace{-0.5cm}
\end{figure}

For simplicity, we set the input power $u_1 = 1$, and consider only the needle depth $u_2$ as the free control input. We can express system \eqref{eq:pde sys} as affected by a CDM $P$ as:
\begin{equation*}
	\dot{x}(t, \xi) = \begin{bmatrix}
		\dot{z}(t, \xi) \\
		\dot{d}(t)
	\end{bmatrix}
	=
	\begin{bmatrix}
		a \nabla^2 x_1(t, \xi) \\
		0
	\end{bmatrix}
	+
	\begin{bmatrix}
		q(\xi) & 0 \\
		0 & 1
	\end{bmatrix}
	P
	\begin{bmatrix}
		u_1 \\ u_2
	\end{bmatrix}.
\end{equation*}
We consider a CDM of the form
$
	P u =
%	\begin{bmatrix}
%	P_{11} & P_{12} \\
%	P_{21} & P_{22}
%	\end{bmatrix}
%	u
%	=
	\begin{bmatrix}
	I & P_{12} \\
	0 & I
	\end{bmatrix}
	u,
$
where $P_{12}$ is the map to be identified. We are interested in a 3-mode piecewise linear CDM $P_{12}$, with $\check{U}_1 = [0, 0.25]$, $\check{U}_2 = [0.5, 0.75]$, and $\check{U}_3 = [0.75, 1]$; these regions are illustrated in Fig.~\ref{fig:pde-ex}. Region 1 corresponds to a charred region at the top of the tissue where the needle does not fully contact the tissue. Region 2 is a layer of pristine tissue, where the original dynamics act. Region 3 is a layer of highly vascularized tissue, in which a large fraction of heat that is added to the system gets transported away. We consider a piecewise linear function
\begin{equation*}
\begin{split}
	P_{12} p := \ &(0.25 + 3p) \llbracket p < 0.25 \rrbracket + \llbracket 0.25 \leq p \leq 0.75 \rrbracket \\
	+ \ &(2.5 - 2p) \llbracket p > 0.75 \rrbracket.
\end{split}
\end{equation*}
We consider a sinusoidal control signal for the probe depth with a period of 0.3 seconds, $u_2(t) = (1-\cos(20 \pi t / 3))/2$, and a state--input sampling frequency of 20 Hz. We assume stochastic sampling periods, where the time is perturbed with a uniform 0.01 second error to model signal processing delays. The underlying goal of this application is to perform passive probing of the affected tissue layers and reconfigure the thermodynamics model to account for tissue damage, as is commonly encountered in electrosurgery. Fig.~\ref{fig:dh-eps} shows on the left the Hausdorff distance between each affected region over time to show that approximations become tighter with time, according to the decreasing minimal covering radius $\epsilon$ (right), as shown in Lemma~\ref{lm:convergence of approximations of U}. After three samples in each region, we uniquely identify the appropriate affine map, but the inner-approximation of the affected region is refined passively over time.
%, without any changes to the control input being made, much as discussed in Corollary~\ref{cor:approximation error convergence}.
%Persistent excitation conditions to ensure that $\epsilon$ decreases with time will be the subject of future work.

\begin{figure}[t]
	\centering
    \includegraphics[width=0.49\linewidth]{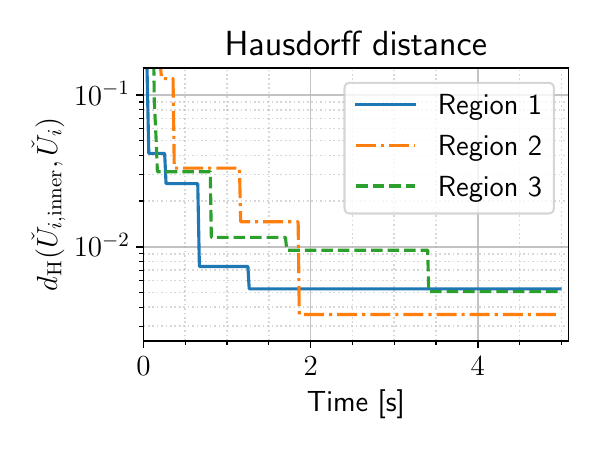}
    \includegraphics[width=0.49\linewidth]{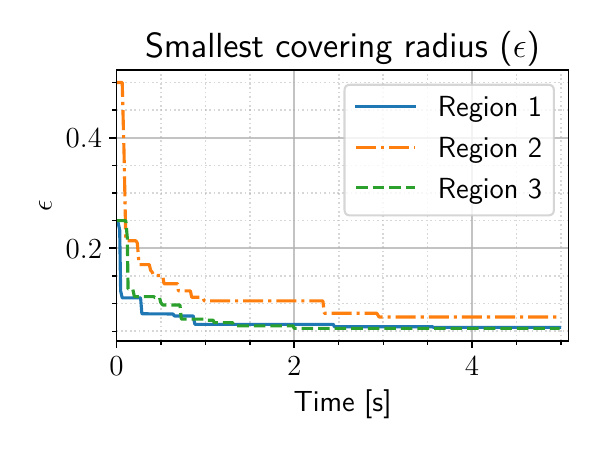}
    \caption{Hausdorff distance error for the inner approximations (left) and minimal covering radius (right) of the three affected regions on which the c-CDM of \eqref{eq:pde sys} acts as a function of time.}
    \label{fig:dh-eps}
    \vspace{-0.5cm}
\end{figure}

%We take $\check{P}_1 u = (2.5 + 0.5 u_1, 0.125 + 0.5 u_2)$, $\check{P}_1 u = (9 - 0.2 u_1, u_2)$

\section{Conclusion}\label{sec:conclusion}

In this work, we have introduced the concept of a \emph{control authority degradation map} (CDM). We have proved conditions on the identifiability of a broad class CDMs, including $N$-mode affine CDMs and Lipschitz continuous CDMs, for a class of affine-in-control nonlinear systems. Based on the identifiability results, we have formulated a constructive method for reconstruction or approximating CDMs, with explicit bounds on the approximation error. Our CDM identification method is executable in real time, and is guaranteed to monotonically decrease in error as more full-state observations become available. 
%Having obtained an approximation of the CDM, we introduced the concept of \emph{viabilizing remapping}, wherein commanded control signals are mapped to viabilizing control inputs, which approximate the command control signal after degradation. We have given conditions for which there exist Lipschitz continuous viabilizing maps, and we provided an efficient constructive approach to obtaining viabilizing maps. 
We apply our methods of CDM identification and viabilization of control signals to a controlled partial differential equation motivated by an electrosurgical process, showing how our guaranteed CDM reconstruction quality improves over time.

\bibliographystyle{IEEEtran}
\bibliography{root}

\begin{thebibliography}{10}
\providecommand{\url}[1]{#1}
\csname url@samestyle\endcsname
\providecommand{\newblock}{\relax}
\providecommand{\bibinfo}[2]{#2}
\providecommand{\BIBentrySTDinterwordspacing}{\spaceskip=0pt\relax}
\providecommand{\BIBentryALTinterwordstretchfactor}{4}
\providecommand{\BIBentryALTinterwordspacing}{\spaceskip=\fontdimen2\font plus
\BIBentryALTinterwordstretchfactor\fontdimen3\font minus
  \fontdimen4\font\relax}
\providecommand{\BIBforeignlanguage}[2]{{%
\expandafter\ifx\csname l@#1\endcsname\relax
\typeout{** WARNING: IEEEtran.bst: No hyphenation pattern has been}%
\typeout{** loaded for the language `#1'. Using the pattern for}%
\typeout{** the default language instead.}%
\else
\language=\csname l@#1\endcsname
\fi
#2}}
\providecommand{\BIBdecl}{\relax}
\BIBdecl

\bibitem{Blanke2006}
M.~Blanke, M.~Kinnaert, J.~Lunze, and M.~Staroswiecki, \emph{Diagnosis and
  Fault-Tolerant Control}.\hskip 1em plus 0.5em minus 0.4em\relax {Berlin,
  Germany}: {Springer Berlin Heidelberg}, 2006.

\bibitem{Mo2016}
H.~Mo and M.~Xie, ``A dynamic approach to performance analysis and reliability
  improvement of control systems with degraded components,'' \emph{IEEE
  Transactions on Systems, Man, and Cybernetics: Systems}, vol.~46, no.~10, pp.
  1404--1414, 2016.

\bibitem{Si2020}
X.~Si, Z.~Ren, X.~Hu, C.~Hu, and Q.~Shi, ``A novel degradation modeling and
  prognostic framework for closed-loop systems with degrading actuator,''
  \emph{IEEE Transactions on Industrial Electronics}, vol.~67, no.~11, pp.
  9635--9647, 2020.

\bibitem{Wang2014a}
Z.~Wang, M.~Rodrigues, D.~Theilliol, and Y.~Shen, ``Fault-tolerant control for
  discrete linear systems with consideration of actuator saturation and
  performance degradation,'' in \emph{19th {{IFAC World Congress}}}, {Cape
  Town, South Africa}, 2014, pp. 499--504.

\bibitem{Niu2022}
R.~Niu, S.~M. Hassaan, and S.~Z. Yong, ``A multi-parametric method for active
  model discrimination of nonlinear systems with temporal logic-constrained
  switching,'' in \emph{2022 {{American Control Conference}}}, {Atlanta, GA,
  USA}, 2022, pp. 1652--1658.

\bibitem{Munkres2000}
J.~R. Munkres, \emph{Topology}, 2nd~ed.\hskip 1em plus 0.5em minus 0.4em\relax
  {Upper Saddle River, USA}: {Prentice Hall}, 2000.

\bibitem{Rockafellar1970}
R.~T. Rockafellar, \emph{Convex {{Analysis}}}.\hskip 1em plus 0.5em minus
  0.4em\relax {Princeton, NJ, USA}: {Princeton University Press}, 1970.

\bibitem{Richter2011}
J.~H. Richter, \emph{Reconfigurable {{Control}} of {{Nonlinear Dynamical
  Systems}}}.\hskip 1em plus 0.5em minus 0.4em\relax {Berlin, Germany}:
  {Springer Berlin Heidelberg}, 2011.

\bibitem{El-Kebir2022e}
H.~{El-Kebir}, A.~Pirosmanishvili, and M.~Ornik, ``Online guaranteed reachable
  set approximation for systems with changed dynamics and control authority,''
  \emph{arXiv:2203.10220 [math.OC]}, 2022.

\bibitem{Nashed1976}
M.~Z. Nashed and G.~F. Votruba, ``A {{Unified Operator Theory}} of
  {{Generalized Inverses}},'' in \emph{Generalized {{Inverses}} and
  {{Applications}}}, M.~Z. Nashed, Ed.\hskip 1em plus 0.5em minus 0.4em\relax
  {Cambridge, MA, USA}: {Academic Press}, 1976, pp. 1--109.

\bibitem{El-Kebir2022d}
H.~{El-Kebir}, J.~Ran, Y.~Lee, L.~P. Chamorro, M.~{Ostoja-Starzewski},
  R.~Berlin, and J.~Bentsman, ``Minimally invasive live tissue high-fidelity
  thermophysical modeling using real-time thermography,'' \emph{IEEE
  Transactions on Biomedical Engineering}, 2022, (early access).

\end{thebibliography}
%\bibliography{root.bib}


%\end{IEEEbiography}
%% insert where needed to balance the two columns on the last page with
%%\newpage


%\vfill



\end{document}